\documentclass[10pt,onecolumn]{IEEEtran}
\usepackage[noadjust]{cite}
\usepackage[english]{babel}
\usepackage{amsmath,amssymb,amsthm,mathrsfs}
\usepackage{amsfonts}
\usepackage{caption}
\usepackage{flushend}
\usepackage{epsf}
\usepackage{graphics}
\usepackage[final,hyperfootnotes=false]{hyperref}
\usepackage{enumerate}
\usepackage{url}
\usepackage{tikz}
\usetikzlibrary{arrows,positioning,fit,backgrounds}
\usetikzlibrary{calc}

\allowdisplaybreaks

\newtheorem{thm}{Theorem}
\newtheorem{cor}[thm]{Corollary}

\newtheorem{lem}[thm]{Lemma}
\newtheorem{prop}[thm]{Proposition}

\theoremstyle{definition}

\theoremstyle{remark}
\newtheorem{rem}{Remark}

\newcommand{\mb}{\mathbf}

\newcommand{\half}{\tfrac{1}{2}}

\DeclareMathOperator*{\rmd}{d}

\newcommand{\liu}[1]{\textcolor{blue}{[Liu: #1]}}
\newcommand{\ayfer}[1]{\textcolor{red}{[Ayfer: #1]}}

\usepackage{xcolor}

\begin{document}

\title{Capacity Upper Bounds for the Relay Channel via Reverse Hypercontractivity
}

\author{Jingbo Liu, \quad Ayfer Ozgur\\
{jingbo}@mit.edu, Institute for Data, Systems, and Society, Massachusetts Institute of Technology\\
{aozgur}@stanford.edu, Electrical Engineering, Stanford University\\
}

\maketitle

\begin{abstract}
The primitive relay channel, introduced by Cover in 1987, is the simplest single-source single-destination network model that captures some of the most essential features and challenges of relaying in wireless networks. Recently, Wu and Ozgur developed upper bounds on the capacity of this channel that are tighter than the cutset bound. In this paper, we recover, generalize and improve  their upper bounds with simpler proofs that rely on a converse technique recently introduced by  Liu, van Handel and Verd\'u that builds on reverse hypercontractivity. To our knowledge, this is the first application of reverse hypercontractivity for proving first-order converses in network information theory.
\end{abstract}


\section{Introduction}
The primitive relay channel, introduced by Cover in 1987 \cite{cover} models the communication scenario where a source-destination pair is assisted by a single relay which is connected to the destination with an independent channel of some finite capacity \cite{cover}. The source's input $X$ is received by the relay $Z$ and the destination $Y$ through a channel  $p(y,z|x)$, and the relay $Z$ can communicate to the destination $Y$ via an error-free digital link of rate $C_0$. See Figure~\ref{F:primitive}. As noted by Kim \cite{Kim}, \emph{the primitive relay channel is the simplest model that captures the most essential features and challenges of relaying in wireless networks}. It can be regarded as the simplest network model that intertwines channel coding with source coding. \emph{On the one hand, it is the simplest channel coding problem (from the source transmitter's point of view) with a source coding constraint; on the other hand, it is the simplest source coding problem (from the relay's point of view) for a channel code.} As such, even-though its capacity remains unknown, it has served as a good testbed for developing new relay coding schemes as well as new converse techniques  over the last three decades \cite{Zhang, Kim, mod,  Ulukus, muriel, Xue, WuOzgurXie, WuOzgur,WuOzgur-general, WuBarnesOzgur, WuBarnesOzgur-general, natasha}). \footnote{Strictly speaking, Cover \cite{cover} considers a special case of the primitive relay channel in Figure~\ref{F:primitive} where the channels from $X$ to $Y$ and $Z$ are independent and symmetric, i.e. $Y$ and $Z$ are conditionally independent given $X$ and  $P_{Y|X}=P_{Z|X}$, where the first assumption appears to be more important than the second. The more general  formulation in Figure~\ref{F:primitive} and the term ``primitive'' appears in \cite{Kim}. However, many of the papers we overview next as well as the current paper consider Cover's original formulation. In a sense, Cover's  formulation is the simplest setting to which Kim's comment applies. It is the simplest setting which combines channel and source coding and as such encapsulates most of the difficulties in a relay channel. This as also evidenced by the fact that the capacity of this simplest setting remains open since Cover, 1987.}
\begin{figure}[hbt]
\centering
\includegraphics[width=0.3\textwidth]{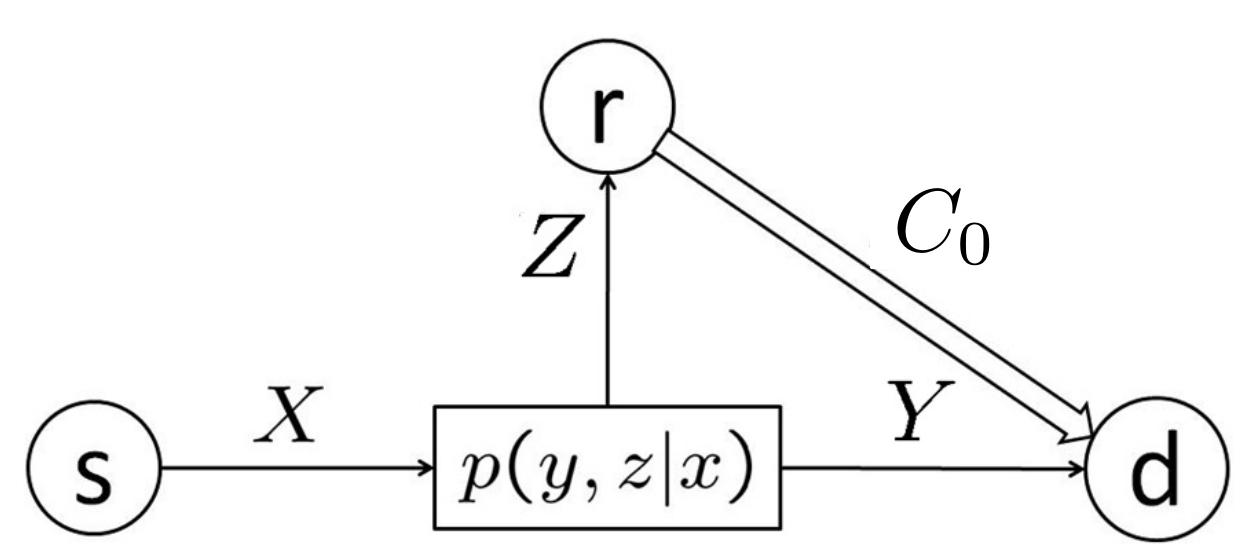}
\caption{Primitive relay channel.}
 \label{F:primitive}
\end{figure}

The classical upper bound on the capacity of the  relay channel, as well as its primitive version, is the so-called cutset bound developed by Cover and El Gamal in 1979 \cite{CoverElgamal}. It bounds the capacity of a network by its minimal cut capacity, in a flavor similar to the famous max-flow min-cut theorem  for graphical networks \cite{FF}.
Athough the cutset bound is known to be tight for most special classes of relay channels for which the capacity is known, such as the physically degraded and reversely degraded  \cite{CoverElgamal}, orthogonal \cite{orthogonal}, semi-deterministic \cite{semideterministic}, and deterministic
\cite{deterministic} relay channels,  its suboptimality in the general case has been established in the context of the primitive relay channel.
In particular, when the channel in Figure~\ref{F:primitive} is discrete and $Y$ and $Z$ are conditionally independent given $X$, and $Y$ is a stochastically degraded version of $Z$ , Zhang \cite{Zhang} uses the blowing-up lemma \cite{Marton} to show that the capacity of this channel can be strictly smaller than the cutset bound in certain regimes. However, Zhang's result is limited to proving that the inequality between the capacity and the cutset bound can be strict in certain cases and does not provide an explicit bound (tighter than the cutset bound) which can be used to benchmark the performance of achievable schemes.  For a special case of the discrete primitive relay channel where the noise for the $X$-$Y$ link is modulo additive and $Z$ is a corrupted version of this noise,  Aleksic, Razaghi and Yu characterize the capacity and show that it is strictly lower than the cutset bound \cite{mod} (see also \cite{Ulukus}). While this result provides an exact capacity characterization for a non-trivial special case, it builds strongly on the peculiarity of the channel model and in this respect its scope is more limited than Zhang's result. A stronger upper bound for Zhang's setting has been established by Xue in \cite{Xue}, where the gap to the cutset bound is related to the reliability function of the  $X$-$Y$ link. More recently, Wu, Ozgur and Xie in \cite{WuOzgurXie} and Wu and Ozgur in \cite{WuOzgur} develop new upper bounds on the capacity of primitive relay channel focusing on the case when the channels from $X$ to $Y$ and $Z$ are independent and symmetric (as in Cover's original formulation in \cite{cover}). In particular, \cite{WuOzgur}  considers the Gaussian  case, where $Y$ and $Z$ are obtained by passing $X$ through independent symmetric Gaussian channels (see \cite{WuOzgur-general} for an extension to the asymmetric case), and  develop the first upper bounds that are tighter than the cutset bound for this canonical model for wireless relay networks. For discrete channels, the bounds in \cite{WuOzgur} significantly improve on the earlier bound in \cite{Xue}. 
One should note that all these bounds  \cite{WuOzgurXie, WuOzgur,WuOzgur-general, Zhang, Xue} rely on the same key ingredient, concentration of measure and specifically the  blowing-up lemma  in either discrete or Gaussian space  \cite{Marton, Talagrand}, even though each of these works  develop a new and different argument for invoking this ingredient; for example, the earlier arguments for discrete channels in \cite{Zhang, Xue} do not extend to continuous, in particular Gaussian, channels considered in  \cite{WuOzgur}.

In this paper, we prove upper bounds on the capacity of the primitive relay channel by using a different technique first  introduced by  Liu, van Handel and Verd\'u in \cite{ISIT_lhv2017_no_url} which relies on the reverse hypercontractivity of Markov semigroups. 
In a unified approach (by using different Markov semigroups for different channel models),
we prove sharper bounds for Gaussian channels and channels with bounded density (which includes all discrete memoryless channels) than previous counterparts in the Gaussian \cite{WuOzgur}  and the discrete memoryless cases \cite{WuOzgurXie}.
The reverse hypercontractivity approach has been integrated with the functional representations of information measures to successfully resolve the second-order scaling of certain multiuser information theory problems in the nonvanishing error regime \cite{lcv_onecomm,ISIT_lccv_smooth2016,ISIT_lhv2017_no_url,liu_thesis,liu18}. 
We remark that  reverse hypercontractivity has gained some recent interest within the information theory community.
For example \cite{kamath15}\cite{lccv16}\cite{beigi16} studied  equivalent formulations for reverse hypercontractivity and \cite{nair_e}
computed the reverse hypercontractivity region for the erasure channel. However, the application of reverse hypercontractivity to obtain bounds on the fundamental limits of operational problems has been limited beyond \cite{ISIT_lhv2017_no_url}, and  to the best of our knowledge the current paper provides the first application of reverse hypercontractivity to derive  bounds on the (first-order) capacity of multi-user networks.\footnote{Polyanskiy and Wu  use the transportation-cost inequality to prove a first-order converse for the interference channel in \cite{polyanskiy2016wasserstein}; since transportation-cost inequalities are known to imply concentration of measure and the blowing-up lemma in particular, this result can be thought of as an indirect application of concentration of measure to information-theoretic converses.} 

We note that in more recent work \cite{WuBarnesOzgur}, Wu, Barnes and Ozgur develop a tighter upper bound on the capacity of the Gaussian primitive relay channel which significantly improves on \cite{WuOzgur}  and is strong enough to resolve the open problem posed by Cover in \cite{cover}, which he calls ``The Capacity of the Relay Channel''. Their proof builds on a rearrangement inequality on the sphere \cite{baernstein}. It would be interesting to see if this stronger result can be recovered with simpler proofs based on the reverse hypercontractivity approach we develop in the current paper.

The paper is organized as follows. 
Section~\ref{sec_prob} reviews the precise formulation of the primary relay channel problem.
In Section~\ref{sec_pre} we recall the reverse hypercontractivity results for the Ornstein-Uhlenbeck and the semi-simple semigroup we use in this paper. 
The main results are presented in Section~\ref{sec_main} and proved in Section~\ref{sec_proof}.

\section{Problem Formulation}\label{sec_prob}
Consider a symmetric primitive relay channel as depicted in Fig. \ref{F:primitive}. The source's input $X\in\mathcal{X} $ is received by the relay $Z\in\Omega$ and the destination $Y\in\Omega$  through symmetric and independent channels $W:=P_{Y|X}=P_{Z|X}$ and $P_{Y,Z|X}=P_{Y|X}P_{Z|X}$. The relay $Z$ can communicate to the destination $Y$ via an error-free digital link of rate $C_0$ nats/ channel use.

For this channel, a code of rate $R$ and blocklength $n$, denoted by $$(\mathcal{C}_{(n,R)}, f_n(z^n), g_n(y^n,f_n(z^n))), \mbox{ or simply, } (\mathcal{C}_{(n,R)}, f_n, g_n), $$
consists of the following:
\begin{enumerate}
  \item A codebook at the source $X$: $\mathcal{C}_{(n,R)}=\{x^n (m)\in\mathcal{X}^n, m\in \{1,2,\ldots, \lceil e^{nR}\rceil\} \};$
  \item An encoding function at the relay $Z$: $f_n: \Omega^n \rightarrow \{1,2,\ldots, \lceil e^{nC_0}\rceil \};$
  \item A decoding function at the destination $Y$: $g_n: \Omega^n  \times \{1,2,\ldots, \lceil e^{nC_0}\rceil\}  \rightarrow \{1,2,\ldots, \lceil e^{nR}\rceil\}.$
\end{enumerate}

The average probability of error of the code is defined as
\begin{align}
P_e^{(n)}=\mbox{Pr}(g_n(Y^n,f_n(Z^n)) \neq M ),
\label{e_pe}
\end{align}
where the message $M$ is assumed to be uniformly drawn from the message set $ \{1,2,\ldots, \lceil e^{nR}\rceil\}$. A rate $R$ is said to be achievable if there exists a sequence of codes
$$\{(\mathcal{C}_{(n,R)}, f_n, g_n)\}_{n=1}^{\infty}$$
such that the average probability of error $P_e^{(n)} \to 0$ as $n \to \infty$. The capacity of the primitive relay channel is the supremum of all achievable rates, denoted by $C(C_0)$.

The following proposition summarizes an intermediate step in the derivation of the cutset bound \cite{CoverElgamal}.
\begin{prop}\label{main:prop}
Consider a symmetric primitive relay channel where $W:=P_{Y|X}=P_{Z|X}$.
Let $n$ be a positive integer and $P_{Y^n|X^n=x^n}=P_{Z^n|X^n=x^n}=W_{x^n}:=\otimes_{i=1}^n W_{x_i}$.
Suppose that there exist encoding and decoding schemes with error probability  $\epsilon:=P_e^{(n)}$ (see definition in \eqref{e_pe}).
Let $Q$ be equiprobable on $\{1,\dots,n\}$ and independent of $(X^n,Y^n,Z^n)$.
Then 
\begin{align}
 R  & \leq I(X_Q;Y_Q,Z_Q)+\mu(\epsilon)  \label{eq1} \\
R   & \leq  I(X_Q;Y_Q)+\frac{1}{n}H(I|Y^n) -\frac{1}{n}H(I|X^n) +\mu(\epsilon)  \label{eq2} 
\end{align}
where $\mu(\epsilon)\to 0$ as $\epsilon\to 0$,
$X^n$ is the random codeword,
and $P_{Y^n|X^n=x^n}=P_{Z^n|X^n=x^n}=W_{x^n}:=\otimes_{i=1}^n W_{x_i}$,
$I=f(Z^n)-Z^n-X^n-Y^n$.
\end{prop}
For completeness we include the short proof here:
\begin{proof}
Using the Fano inequality and the  chain rules, we have
\begin{align}
 R&\le \frac{1}{n}I(X^n;I,Y^n)+\mu(\epsilon)
\\
&\le \frac{1}{n}I(X^n;Z^n,Y^n)+\mu(\epsilon)
\\
&= \frac{1}{n}\sum_{i=1}^nI(X_i;Z_i,Y_i)-\frac{1}{n}D(P_{Y^nZ^n}\|\otimes_{i=1}^nP_{Y_iZ_i})+\mu(\epsilon)
\\
&\le \frac{1}{n}\sum_{i=1}^nI(X_i;Z_i,Y_i)+\mu(\epsilon)\label{e6}
\\
&=I(X_Q;Z_Q,Y_Q|Q)+\mu(\epsilon)
\end{align}
where $D(\cdot\|\cdot)$ denotes the relative entropy;
and 
\begin{align}
 R&\le 
 \frac{1}{n}I(X^n;I,Y^n)+\mu(\epsilon)
 \\
 &=\frac{1}{n}\sum_{i=1}^nI(X^n;Y^n)+\frac{1}{n}\sum_{i=1}^nH(I|Y^n)-\frac{1}{n}\sum_{i=1}^nH(I|X^n)+\mu(\epsilon)
 \\
 &=\frac{1}{n}\sum_{i=1}^nI(X_i;Y_i)-\frac{1}{n}D(P_{Y^n}\|\otimes_{i=1}^nP_{Y_i})+\frac{1}{n}\sum_{i=1}^nH(I|Y^n)-\frac{1}{n}\sum_{i=1}^nH(I|X^n)+\mu(\epsilon)
 \\
 &=\frac{1}{n}\sum_{i=1}^nI(X_i;Y_i)+\frac{1}{n}\sum_{i=1}^nH(I|Y^n)-\frac{1}{n}\sum_{i=1}^nH(I|X^n)+\mu(\epsilon).
 \\
 &=I(X_Q;Y_Q|Q)+\frac{1}{n}\sum_{i=1}^nH(I|Y^n)-\frac{1}{n}\sum_{i=1}^nH(I|X^n)+\mu(\epsilon).
\end{align}
The proof is completed since the Markov chain $Q-X_Q-Y_Q$ implies that $I(X_Q;Y_Q|Q)\le I(X_Q;Y_Q)$; similarly $I(X_Q;Z_Q,Y_Q|Q)\le I(X_Q;Z_Q,Y_Q)$.
\end{proof}

Note that this is an $n$-letter  upper bound on the capacity of the channel. However if we can establish bounds on $\frac{1}{n}H(I|Y^n)$ for any value of $\frac{1}{n}H(I|X^n):=h$ satisfying the conditions in the proposition, we can use the proposition to establish an explicit computable upper bound on the capacity of the symmetric primitive relay channel. This is the approach of \cite{WuOzgurXie,WuOzgur,WuBarnesOzgur}, which we also adopt in this paper.

\section{Preliminaries}\label{sec_pre}
In this section, we introduce some more notation and provide a brief overview of reverse hypercontractivity, which will be our main tool for proving upper bounds on the capacity of the relay channel defined in the previous section. 
\subsection {Notation}
Given a measurable space $\mathcal{Y}$,
let $\mathcal{H}_{[0,1]}(\mathcal{Y})$ (resp.\ $\mathcal{H}_+(\mathcal{Y})$) be the set of measurable functions on $\mathcal{Y}$ taking values in $[0,1]$ (resp.\ $[0,\infty)$). Given a probability measure $Q$ on $\mathcal{Y}$,  $f\in\mathcal{H}_+(\mathcal{Y})$, and $p\in (0,\infty)$, let
\begin{align}
Q(f)&:=\int f\rmd Q,
\\
\|f\|_{L^p(Q)}=\|f\|_p&:=[Q(f^p)]^{\frac{1}{p}}.
\end{align}
Then by a limiting argument we have
\begin{align}
\|f\|_{L^0(Q)}:=e^{Q(\ln f)}.
\end{align}
Given a channel (i.e.\ conditional probability) $W=P_{Y|X}$, we will often write $W_{x^n}:=\otimes_{i=1}^nP_{Y|X=x_i}$. Finally, the bases in all logarithms, exponentials and information-theoretic quantities in this paper are natural.

\subsection{Reverse hypercontractivity}
Let $T\colon\mathcal{H}_+(\mathcal{Y})\to\mathcal{H}_+(\mathcal{Y})$ be a nonnegativity-preserving map (i.e., nonnegative functions are mapped to nonnegative functions) and let $Q$ be a fixed reference measure. $T$ is said to be reverse hypercontractive (see for example \cite{borell1982positivity}) if for some $0\le p<q\le 1$,
\begin{align}
\|Tf\|_p\ge \|f\|_q,
\quad
\forall f\in\mathcal{H}_+(\mathcal{Y}).
\label{e4}
\end{align}
Note that if $T$ is a conditional expectation operator (i.e.,
there exists some conditional probability $P_{Y|X}$ such that \begin{align}
(Tf)(x)=P_{Y|X=x}(f)
\end{align}
for each $x\in \mathcal{Y}$), then \eqref{e4} holds if $0\le p=q\le 1$, by Jensen's inequality.
Reverse hypercontractivity characterizes how $T$ is able to increase the small values of the function (i.e.\ positivity improving \cite{borell1982positivity}).
Markov semigroups provide a rich source of operators satisfying reverse hypercontractivity (among other favorable properties).

\subsection{Ornstein-Uhlenbeck semigroups}
For any $x^n\in\mathbb{R}^n$ and $t\ge 0$, define a linear operator $T_{x^n,t}$ by 
\begin{align}
T_{x^n,t}f(y^n)
:=\mathbb{E}[f(e^{-t}y^n+(1-e^{-t})x^n
+\sqrt{1-e^{-2t}}V^n)]
\label{e_tx}
\end{align}
for any $f\in \mathcal{H}_+(\mathbb{R}^n)$, where
$V^n\sim\mathcal{N}(0^n,\mb{I}_n)$.
Then $(T_{x^n,t})_{t\ge 0}$ is called the Ornstein-Uhlenbeck (OU) semigroup with stationary measure
$P_{Y^n|X^n=x^n}:=\mathcal{N}(x^n,\mb{I}_n)$ (see for example \cite{borell1982positivity}).
The OU semigroup is among the first examples where a reverse hypercontractivity estimate has been worked out:
\begin{lem}[\cite{borell1982positivity}]\label{lem_ou}
For any $q<p<1$ and $t\ge \half\ln\frac{1-q}{1-p}$,
we have
\begin{align}
\|T_{x^n,t} f\|_q\ge \|f\|_p,\quad\forall f\in \mathcal{H}_+(\mathbb{R}^n)
\end{align}
where the norms are with respect to the stationary measure $\mathcal{N}(x^n,\mb{I}_n)$.
\end{lem}
In particular, taking $q=0$, we see that for any $f\in\mathcal{H}_{[0,1]}(\mathbb{R}^n)$,
\begin{align}
\mathbb{E}[\ln T_{x^n,t}f]
&\ge \|f\|_{1-e^{-2t}}
\\
&\ge
 \frac{1}{1-e^{-2t}}\ln\mathbb{E}[f]
 \label{e8}
 \\
&\ge
\left(1+\frac{1}{2t}\right)\ln\mathbb{E}[f]
,
\quad \forall f\in \mathcal{H}_{[0,1]}(\mathbb{R}^n)
\label{e_q0}
\end{align}
where \eqref{e8} used the fact that $f^p\ge f$,
\eqref{e_q0} follows from $\frac{1}{1-e^{-2t}}\le \frac{1}{2t}+1$ (note that $\ln\mathbb{E}[f]$ is negative!),
and the expectations are with respect to the stationary measure.

We remark that Lemma~\ref{lem_ou} is completely dual to the (forward) hypercontractivity estimate for the OU semigroup (see e.g.\ \cite{gross1975logarithmic}) in the sense that the dependence of the parameters $t\ge \half\ln\frac{1-q}{1-p}$ takes on the same formula (although in the case of hypercontractivity, both $p$ and $q$ are greater than 1).
However, this is merely a coincidence for the OU semigroup.
The reverse hypercontractivity is generally weaker (and hence more common) than hypercontractivity \cite{mossel2013reverse}, as the next example illustrates.

\subsection{Simple and semi-simple semigroups}
The simplest and the most natural semigroup that can be defined for any given stationary measure $P$ on a measurable space $\mathcal{Y}$ is the \emph{simple semigroup} (see e.g.\ \cite{mossel2013reverse}),
defined by
\begin{align}
T_tf= e^{-t}f+(1-e^{-t})P(f),\quad\forall t\ge 0, f\in\mathcal{H}_+(\mathcal{Y}).
\end{align}
In the i.i.d.\ case, the tensor product $T_t^{\otimes n}$ of any Markov semigroup operator $T_t$ forms a new Markov semigroup whose stationary measure is $P^{\otimes n}$.
The product inherits the same reverse hypercontractive inequality as its factors, 
which is called \emph{tensorization} (see e.g.\ \cite{mossel2013reverse}).
We call the tensor product of a simple semigroup a \emph{semi-simple} semigroup. We will use the semi-simple semigroup in the case of  discrete memoryless channels and continious channels with bounded density.


By establishing the equivalence between the modified log-Sobolev inequality and the reverse hypercontractivity,
the recent paper by Mossel et al.\ \cite{mossel2013reverse} established the following striking universal reverse hypercontractivity estimate which does not depend on the stationary measure $P$.
We remark that the (forward) hypercontractivity is drastically different in that the 
bound depends on the smallest probability mass in $P$ (see e.g.\ \cite{boucheron2004concentration}).
\begin{lem}[\cite{mossel2013reverse}]\label{lem_mos}
Let $P_i$ be a probability distribution on $\mathcal{Y}$ for each $i=1,\dots,n$,
and let
\begin{align}
T_t:=\otimes_{i=1}^n[e^{-t}+(1-e^{-t})P_i]
\end{align}
be a Markov semigroup with stationary measure $\otimes_{i=1}^nP$.
For any $q<p<1$ and $t\ge \ln\frac{1-q}{1-p}$,
we have
\begin{align}
\|T_t f\|_q\ge \|f\|_p,\quad\forall f\in \mathcal{H}_+(\mathcal{Y}^n)
\end{align}
where the norms are with respect to the stationary measure.
\end{lem}
In particular, taking $q=0$ and using the same arguments before, we obtain
\begin{align}
\mathbb{E}[\ln T_tf]
&\ge
 \frac{1}{1-e^{-t}}\ln\mathbb{E}[f]
 \\
&\ge
\left(1+\frac{1}{t}\right)\ln\mathbb{E}[f]
,
\quad \forall f\in \mathcal{H}_{[0,1]}(\mathcal{Y}^n)
\end{align}
where the expectations are with respect to the stationary measure.

\section{Main Results}\label{sec_main}
We first state our results for the Gaussian case and then for channels with bounded density.

\subsection{Gaussian channels}\label{sec_gfano}
\begin{lem}\label{lem2}.
Let $I-Z^n-X^n-Y^n$ and $P_{Y^n|X^n=x^n}=P_{Z^n|X^n=x^n}=\mathcal{N}(x^n,\mb{I}_n)$.
Define $h:=\frac{1}{n}H(I|X^n)$. Then
\begin{align}
H(I|Y^n)
&\le
n\min_{t>0}\left\{t+\frac{1}{1-e^{-2t}}h
\right\}
\\
&=\frac{n}{2}\ln\left(
1+h+\sqrt{h^2+2h}
\right)
+\frac{n}{2}\left(h+\sqrt{h^2+2h}\right)
\label{e_17}
\\
&\le
n(h+\sqrt{2h}).
\label{e5}
\end{align}
\end{lem}

The relaxed bound \eqref{e5} is the same as \cite[Lemma 4.1]{WuOzgur}.\footnote{Note that the entropy and rates in \cite{WuOzgur} are defined in bits, while we use nats in the current paper.} Thus \eqref{e_17} is a strict improvement of \cite[Lemma~4.1]{WuOzgur}. 
When combined with Proposition~\ref{main:prop}, this results immediately yield the following upper bound on the capacity of the Gaussian symmetric primitive relay channel.
We say $P>0$ is an average power constraint if the codewords satisfy
\begin{align}
 \frac{1}{2^{nR}}\sum_{m=1}^{2^{nR}}\|x^n(m)\|_2^2\le nP.
\end{align}

\begin{cor}\label{cor1}
The capacity of the symmetric Gaussian primitive relay channel with $P_{Y|X=x}=P_{Z|X=x}=\mathcal{N}(0,N)$, average power constraint $P$ 
and relay rate $C_0\ge 0$ satisfies
\begin{align}
 C(C_0)  & \leq 
\left\{ \frac{1}{2}\ln\left(1+\frac{2P}{N}\right),\,
 \frac{1}{2}\ln\left(1+\frac{P}{N}\right)+C_0-c^{-1}(C_0)\right\}
 \end{align}
where $c^{-1}(\cdot)$ denotes the inverse of the surjective function $c\colon [0,\infty)\to[0,\infty),\,h\mapsto \frac{1}{2}\ln\left(
1+h+\sqrt{h^2+2h}
\right)
+\frac{1}{2}\left(h+\sqrt{h^2+2h}\right)$. 
\end{cor}
\begin{proof} Let $X^n$ be an equiprobably selected codeword,
and $Q$ be the time sharing random variable as in Proposition~\ref{main:prop}, 
and note that
\begin{align}\label{eq:powerconst}
E[X^2_Q]=  \frac{1}{n}\sum_{i=1}^{n} E[X^2_i]  =  \frac{1}{n} E\left[\sum_{i=1}^{n}X^2_i\right]  \leq P.
\end{align}
Then Proposition~\ref{main:prop} gives
\begin{align}
 R  & \leq I(X_Q;Y_Q,Z_Q)+\mu(\epsilon)   \\
R   & \leq   I(X_Q;Y_Q)+\frac{1}{n}H(I|Y^n) -\frac{1}{n}H(I|X^n) +\mu(\epsilon)
\\
&\le   I(X_Q;Y_Q)+ \min\{C_0,c(h)\}-h+\mu(\epsilon)
\label{e40}
\end{align}
where we defined $h=\frac{1}{n}H(I|X^n)$, and \eqref{e40} used the fact that $\frac{1}{n}H(I|Y^n)\leq C_0$ and applied  Lemma~\ref{lem2}.
To finish, note that the mutual information terms are maximized by choosing $X_Q\sim\mathcal{N}(0,P)$,
and $\min\{C_0,c(h)\}-h$ is maximized at $h=c^{-1}(C_0)$ which can be seen from the fact that $c(h)-h$ is an increasing function.
\end{proof}
Using similar lines of argument as in the proof of Lemma~\ref{sec_gfano},
we will be able to obtain bounds for channels with bounded densities in Section~\ref{sec_bdd} (by using a different semigroup adapted to that class of channels).
For the same Gaussian setting as Lemma~\ref{sec_gfano},  however,
we can capitalize on certain scaling properties of the Gaussian channel and use slightly different lines of proof,
to show the following sharper bound:
\begin{lem}\label{lem3}
Let $I-Z^n-X^n-Y^n$ where $P_{Y^n|X^n=x^n}=P_{Z^n|X^n=x^n}=\mathcal{N}(x^n,\mb{I}_n)$, and define $h_2:=\frac{1}{n}H(I|Y^n)$ and $h_1:=\frac{1}{n}H(I|X^n)$. Then
\begin{align}
h_2-h_1\le\frac{1}{2}\ln(1+2h_2).
\end{align}
\end{lem}

Asymptotically,  \cite[Lemma 4.1]{WuOzgur} (which is \eqref{e5}) gives\footnote{We write $c(h)\lesssim h$, $h\ll 1$, if $\limsup_{h\downarrow0}\frac{c(h)}{h}\le 1$.}
\begin{align}
  h_2&\lesssim \sqrt{2h_1},\quad h_1\ll 1; \label{e44}
  \\
  h_2&\lesssim 2h_1,\quad h_1\gg 1.\label{e_45}
\end{align}
The bound in \eqref{e_17} has the same asymptotics \eqref{e44}; but with \eqref{e_45} replaced by $h_2\lesssim h_1$, $h_1\gg 1$.
In contrast, the bound in Lemma~\ref{lem3} gives
\begin{align}
  h_2&\lesssim \sqrt{h_1},\quad h_1\ll 1; 
  \\
  h_2&\lesssim h_1,\quad h_1\gg 1,
\end{align}
so Lemma~\ref{lem3} essentially yields improvements by constant factors.
A numerical comparison is shown in Figure~\ref{fig1}.
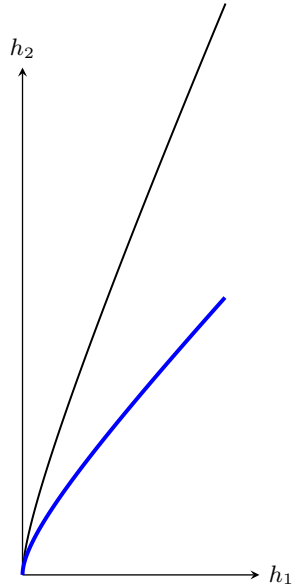
\begin{figure}[ht]
\centering
    \begin{tikzpicture}[scale=0.9,
      dot/.style={draw,fill=black,circle,minimum size=1mm,inner sep=0pt},arw/.style={->,>=stealth}]
      \draw[arw,line width=0.5pt] (0,0) to (0,7.5) node[above,font=\small] {$h_2$};
      \draw[arw,line width=0.5pt] (0,0) to (3.5,0) node[right,font=\small]{$h_1$};
     \draw[line width=0.75pt,domain=0.00001:3-0.00001,variable=\r,samples=100]
       plot({\r},
       {2*\r+sqrt(2*\r)});
     \draw[color=blue,line width=1.5pt,domain=0.00001:4.1-0.00001,variable=\r,samples=100]
       plot({\r-0.5*ln(1+2*\r)},
       {\r});
\end{tikzpicture}
\caption{The bound in \cite[Lemma~4.1]{WuOzgur} (see \eqref{e5}) is plotted in the thin line and the bound in Lemma~\ref{lem3} is plotted in the thick line, for $h_1\in(0,3)$.}
\label{fig1}
\end{figure}


Combined with Proposition~\ref{main:prop}, this results  yields the following upper bound on the capacity of the Gaussian symmetric primitive relay channel. The proof is immediate and follows similar arguments as in the proof of Corollary~\ref{cor1}.

\begin{cor}\label{cor1}
The capacity of the symmetric Gaussian primitive relay channel with $P_{Y|X=x}=P_{Z|X=x}=\mathcal{N}(x,N)$ with average power constraint $P$ satisfies
\begin{align}
 C(C_0)\le\min\left\{\frac{1}{2}\ln\left(1+\frac{2P}{N}\right),\,\frac{1}{2}\ln\left(1+\frac{P}{N}\right)+\frac{1}{2}\ln(1+2C_0)\right\}.
\end{align}
\end{cor}

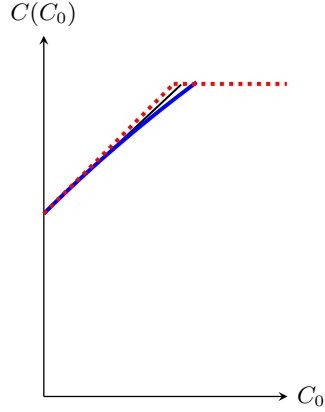
\begin{figure}[ht]
\centering
    \begin{tikzpicture}[scale=12,
      dot/.style={draw,fill=black,circle,minimum size=1mm,inner sep=0pt},arw/.style={->,>=stealth}]
      \draw[arw,line width=0.5pt] (0,0) to (0,0.4) node[above,font=\small] {$C(C_0)$};
      \draw[arw,line width=0.5pt] (0,0) to (0.27,0) node[right,font=\small]{$C_0$};
     \draw[line width=0.75pt,domain=0.00001:0.009-0.00001,variable=\r,samples=100]
       plot({2*\r+sqrt(2*\r)},
       {2*\r+sqrt(2*\r)-\r+0.5*ln(1.5)});       
     \draw[color=blue,line width=1.5pt,domain=0.00001:0.17-0.00001,variable=\r,samples=100]
       plot({\r},
       {0.5*ln(1+2*\r)+0.5*ln(1.5)});
          \draw[color=red,line width=1.5pt,dotted,domain=0.5*ln(4/3):0.27-0.00001,variable=\r,samples=100]
          plot({\r},{0.5*ln(2)});
        \draw[color=red,line width=1.5pt,dotted,domain=0.00001:0.5*ln(4/3)-0.00001,variable=\r,samples=100]
          plot({\r},{\r+0.5*ln(1.5)});
\end{tikzpicture}
\caption{The bound on capacity using in \cite[Lemma~4.1]{WuOzgur} (see \eqref{e5}) is plotted in the thin line.
The bound in Corollary~\ref{cor1} is plotted in the thick line.
The cutset bound is plotted in the dotted line.
The signal to noise ratio is chosen as $\frac{P}{N}=0.5$.
The range of the relay rate is $C_0\in(0,0.27)$.}
\label{fig2}
\end{figure}

\subsection{Channel distributions with bounded densities}\label{sec_bdd}
In this section, we state our results for discrete memoryless channels. More generally, our bounds apply to channels with bounded conditional density, or more precisely when one can find a reference measure such that the density of the output distribution of a stationary memoryless channel can be bounded by a constant independent of the input distribution. 

\begin{lem}\label{lem4}
Fix $W=P_{Y|X}$.
Suppose that $I-Z^n-X^n-Y^n$ where $P_{Y^n|X^n}=P_{Z^n|X^n}=W^{\otimes n}$,
and
\begin{align}
\alpha:=\sup_x\left\|\frac{\rmd P_{Y|X=x}}{\rmd Q_Y}\right\|_{\infty}<\infty
\label{e35}
\end{align}
for some probability measure $Q_Y$.
(In the discrete case, we can always take $\alpha=\sum_{y}\max_{x}W_x(y)$.)
Define $h:=\frac{1}{n}H(I|X^n)$.
Then
\begin{align}
\frac{1}{n}H(I|Y^n)
&\le c_{\alpha}(h):=
\min_{t>0}\left\{(\alpha-1)t+\frac{1}{1-e^{-t}}h
\right\}
\\
&=(\alpha-1)
\left[
\ln\left(1+\frac{h}{2(\alpha-1)}+\sqrt{\frac{h}{\alpha-1}
+\frac{h^2}{4(\alpha-1)^2}}
\right)
+\frac{h}{2(\alpha-1)}
+\sqrt{\frac{h}{\alpha-1}
+\frac{h^2}{4(\alpha-1)^2}}
\right].
\end{align}
\end{lem}
\begin{rem}
In \cite[Lemma 7.1]{WuOzgurXie}, a weaker bound of 
\begin{align}
\frac{1}{n}H(I|Y^n)\le O\left(\sqrt{h}\ln\frac{1}{h}\right),\quad h\to 0
\end{align}
was derived using the blowing-up lemma for discrete memoryless channels. Here we got rid of the a logarithmic factor.
\end{rem}
\begin{rem}\label{e_iinfty}
Recall that the $\infty$-mutual information for a given $P_{Y|X}$ and $P_X$ is defined as 
\begin{align}
 I_{\infty}(X;Y):=\inf_{Q_Y}\ln \left\|\frac{{\rm d}P_{XY}}{{\rm d}(P_X\times Q_Y)}\right\|_{\infty};
\end{align}
see e.g.\ \cite{verdu2015alpha} and the references therein.
Under reasonable regularity conditions (e.g.\ $P_X$ is a fully supported distribution on a countably infinite set), 
we have $ I_{\infty}(X;Y)=\inf_{Q_Y}\ln \sup_x\left\|\frac{{\rm d}P_{Y|X=x}}{{\rm d}Q_Y}\right\|_{\infty}$
and hence the optimal $\alpha$ in \eqref{e35} is $e^{I_{\infty}(X;Y)}$.
Note that the particular choice of $P_X$ is immaterial for the calculation of $I_{\infty}(X;Y)$ as long as $P_X$ is fully supported.
Also note that in the case of Gaussian channel without a power constraint, it is not possible to find $Q_Y$ for which $\alpha$ defined in \eqref{e35} is finite.
\end{rem}

The above lemma immediately yields the following upper bound on the capacity of primitive relay channels. The proof is similar to the proof of Corollary~\ref{cor1}.

\begin{cor}\label{cor2}
Consider a stationary memoryless primitive relay channel with $P_{Y|X=x}=P_{Z|X=x}$ 
and suppose that the condition \eqref{e35} is satisfied. 
Let $\alpha$ and $c_{\alpha}(h)$ be as defined in Lemma~\ref{lem4}. Then the capacity satisfies
\begin{align*}
 C(C_0)  & \leq \min\{I(X;Y,Z),\,
 I(X;Y)+C_0-c_{\alpha}^{-1}(C_0)\}
\end{align*}
for some random variable $X\in\mathcal{X}$ and $c_{\alpha}^{-1}(\cdot)$ denotes the inverse function.
\end{cor}

\section{Proofs}\label{sec_proof}
\subsection{Proof of Lemma~\ref{lem2}}
\begin{proof}
For each integer (relay message) $i\in\mathcal{I}$, let
$f_i\in \mathcal{H}_{[0,1]}(\mathbb{R}^n)$ be the probability of sending $i$ upon observing the output of the relay (in the case of deterministic relay decoder, $f_i$ will be an indicator function).
Then
\begin{align}
\mathbb{P}[I=i|X^n=x^n]=W_{x^n}(f_i).
\end{align}
Hence
\begin{align}
H(I|X^n)=\mathbb{E}_{X^nI}
\left[\ln\frac{1}{W_{X^n}(f_I)}\right].
\label{e45}
\end{align}
But from \eqref{e_q0}, we see that for any $x^n$ and $i$,
\begin{align}
W_{x^n}\left(\ln(T_{x^n,t}f_i)\right)
\ge
\frac{1}{1-e^{-2t}}\ln W_{x^n}(f_i).
\label{e46}
\end{align}
To finish the proof, let us define a convolution operator $T_t$ that does not depend on $x^n$:
\begin{align}
T_{t}f(y^n)
:=\mathbb{E}[f(y^n
+\sqrt{1-e^{-2t}}V^n)].
\label{e_tx}
\end{align}
In words, the action of $T_{x^n,t}$ can be viewed as consisting of two steps: first $T_t$, and then dilate (with center $x^n$) by a factor $e^t$.
It is easy to see a basic fact: when a function is integrated against a measure, the integral is invariant if both the function and the measure contract by the same factor; this observation proves that
\begin{align}
W_{x^n}\left(\ln(T_{x^n,t}f_i)\right)
=\bar{W}_{x^n}\left(\ln(T_tf_i)\right)
\label{e48}
\end{align}
where we defined a new channel $\bar{W}_{x^n}=\mathcal{N}(x^n,e^{-2t}\mb{I}_n)$.
Now define a new Markov chain $I-Z^n-X^n-\bar{Y}^n$ where $P_{\bar{Y}^n|X^n}=\bar{W}$.
We have
\begin{align}
-H(I|\bar{Y}^n)&=
\mathbb{E}_{I\bar{Y}^n}[\ln P_{I|\bar{Y}^n}(I|\bar{Y}^n)]
\\
&\ge
\mathbb{E}_{I\bar{Y}^n}\left[\ln \left((T_tf_I)(\bar{Y}^n)\right)\right]
\label{e_re}
\\
&=\mathbb{E}_{IX^n}\left[\bar{W}_{X^n}\left(\ln (T_tf_I)\right)\right]
\\
&=\mathbb{E}_{IX^n}\left[
W_{X^n}\left(\ln(T_{X^n,t}f_I)\right)
\right]
\label{e_52}
\\
&\ge\frac{1}{1-e^{-2t}}\mathbb{E}_{IX^n}\left[\ln W_{X^n}(f_I)\right]
\label{e_53}
\\
&=-\frac{1}{1-e^{-2t}}H(I|X^n)
\label{e19}
\end{align}
where
\begin{itemize}
\item $\eqref{e_re}$ used the fact that for any $\bar{y}^n$,
\begin{align}
\sum_{i\in\mathcal{I}}(T_tf_i)(\bar{y}^n)
=\left(T_t\sum_{i\in\mathcal{I}}f_i\right)(\bar{y}^n)
=1.
\end{align}
Indeed, upon rearrangements, \eqref{e_re} is reduced to the nonnegativity of the conditional relative entropy
\begin{align}
D(P_{I|\bar{Y}^n}\|Q_{I|\bar{Y}^n}|P_{\bar{Y}^n})\ge 0
\end{align}
where we defined the conditional distribution $Q_{I|\bar{Y}^n}$ as $Q_{I|\bar{Y}^n}(i|\bar{y}^n)=T_tf_i(\bar{y}^n)$, for each $i,\bar{y}^n$.
\item \eqref{e_52} is from \eqref{e48}.
\item \eqref{e_53} is from \eqref{e46} and \eqref{e45}.
\end{itemize}
However,
\begin{align}
H(I|Y^n)-H(I|\bar{Y}^n)
&\le
h(Y^n|I)-h(\bar{Y}^n|I)
\label{e56}
\\
&=I(\sqrt{1-e^{-2t}}G^n;Y^n|I)
\label{e57}
\\
&\le
I(\sqrt{1-e^{-2t}}G^n;\sqrt{1-e^{-2t}}G^n+e^{-t}\bar{G}^n|I)
\\
&\le
nt
\end{align}
where \eqref{e56} can be seen from EPI; In \eqref{e57} we defined $G^n$ and $\bar{G}^n$ to be independent according to $\mathcal{N}(\mb{0},\mb{I}_n)$, and independent of $(I,X^n)$, and put
\begin{align}
\bar{Y}^n&=X^n+e^{-t}\bar{G}^n;
\\
Y^n&=\bar{Y}^n+\sqrt{1-e^{-2t}}G^n.
\end{align}
This constructs a coupling of
\begin{align}
I-X^n-\bar{Y}^n-Y^n
\end{align}
with the desired marginal distributions.
The conclusion then follows by optimizing $t$.
\end{proof}

\subsection{Proof of Lemma~\ref{lem3}}

\begin{proof} The high-level idea of the proof is roughly as follows:
in the proof of Lemma~\ref{lem2} we established the following result: $H(I|\bar{Y}^n)\le\frac{1}{1-e^{-2t}}H(I|X^n)$,
where $\bar{Y}^n$ is obtained by scaling the noise by a factor $e^{-t}$.
Then we bound $H(I|Y^n)$ in terms of $H(I|\bar{Y}^n)$.
If, instead, we somehow scale the noise by a factor of $e^t$ beforehand to cancel this effect, 
then we can directly obtain a bound on $H(I|Y^n)$.

Analogous to \eqref{e48}, define the conditional distribution
\begin{align}
\hat{W}_{x^n}:=\mathcal{N}(x^n,e^{2t}\mb{I}_n),
\end{align}
then we have the following scaling invariance:
\begin{align}
\hat{W}_{x^n}\left(\ln(\hat{T}_{x^n,t}f)\right)
=W_{x^n}\left(\ln(\hat{T}_tf)\right)
\end{align}
where $\hat{T}_{x^n,t}$ is the semigroup with stationary measure $\hat{W}_{x^n}$,
and $\hat{T}_t$ is defined by
\begin{align}
\hat{T}_{t}f(y^n)
:=\mathbb{E}[f(y^n
+e^t\sqrt{1-e^{-2t}}V^n)].
\end{align}
Then we can use the similar steps as before:
\begin{align}
H(I|X^n)
&=\mathbb{E}\left[\ln\frac{1}{W_{X^n}(f_I)}\right]
\\
&\ge
\mathbb{E}\left[\ln\frac{1}{\hat{W}_{X^n}(f_I)}\right]
-nt+\frac{n}{2}(1-e^{-2t})
\label{e_re1}
\\
&\ge(1-e^{-2t})\mathbb{E}
\left[\hat{W}_{X^n}\left(
\ln\frac{1}{\hat{T}_{X^n,t}f_I}
\right)
\right]
-nt+\frac{n}{2}(1-e^{-2t})
\\
&=(1-e^{-2t})\mathbb{E}
\left[W_{X^n}\left(
\ln\frac{1}{\hat{T}_tf_I}
\right)
\right]
-nt+\frac{n}{2}(1-e^{-2t})
\\
&=(1-e^{-2t})\mathbb{E}
\left[
\ln\frac{1}{\hat{T}_tf_I(Y^n)}
\right]
-nt+\frac{n}{2}(1-e^{-2t})
\\
&\ge(1-e^{-2t})H(I|Y^n)
-nt+\frac{n}{2}(1-e^{-2t}).
\end{align}
where
\begin{itemize}
  \item \eqref{e_re1} follows from
  \begin{align}
  \mathbb{E}\left[
  \ln\frac{W_{X^n}(f_I)}{\hat{W}_{X^n}(f_I)}
  \right]
  &=D(P_{I|Z^n}\circ W_{X^n}\|P_{I|Z^n}\circ \hat{W}_{X^n}|P_{X^n})
  \\
  &\le D(W_{X^n}\|\hat{W}_{X^n}|P_{X^n})
  \\
  &=nt-\frac{n}{2}(1-e^{-2t}).
  \end{align}
\end{itemize}
Thus
\begin{align}
\frac{1}{n}I(I;X^n|Y^n)
\le
\inf_{t>0}\left\{t-\frac{1}{2}+(\frac{1}{2}+\frac{1}{n}H(I|Y^n))e^{-2t}\right\}
=\frac{1}{2}\ln\left(1+\frac{2}{n}H(I|Y^n)\right).
\end{align}
\end{proof}
\subsection{Proof of Lemma~\ref{lem4}}
\begin{proof}
As before, define $(f_i)_{i\in\mathcal{I}}$ as the relay decoding probability functions as before. For any $t>0$,
define the linear operators
\begin{align}
T_{x^n,t}&:=\otimes_{i=1}^n(e^{-t}+(1-e^{-t})W_x);
\\
\Lambda_t&:=\otimes_{i=1}^n(e^{-t}+\alpha(1-e^{-t})Q_Y).
\end{align}
Note that since $\alpha\ge 1$, $\Lambda_t$ is not a conditional expectation operator in the sense that it can send the constant 1 function to a nonnegative function exceeding 1 somewhere.
However, we can show that the factor of increase is not to big:
\begin{align}
(\Lambda\cdot 1)(y^n)
&\le
(e^{-t}+\alpha(1-e^{-t}))^n
\\
&\le
e^{(\alpha-1)nt},\quad\forall y^n.
\label{e_notm}
\end{align}
The rest of the proof then follows analogously to the Gaussian case:
\begin{align}
H(I|X^n)
&=\mathbb{E}_{IX^n}\left[\ln\frac{1}{W_{X^n}(f_I)}\right]
\\
&\ge(1-e^{-t})\mathbb{E}_{IX^n}\left[
W_{X^n}\left(
\ln\frac{1}{T_{X^n,t}f_I}
\right)
\right]
\\
&\ge(1-e^{-t})\mathbb{E}_{IX^n}\left[
W_{X^n}\left(
\ln\frac{1}{\Lambda_t f_I}
\right)
\right]
\label{e_dominate}
\\
&\ge
(1-e^{-t})
\mathbb{E}_{IY^n}
\left[
\ln P_{I|Y^n}(I|Y^n)-(\alpha-1)nt
\right]
\label{e_notmuch}
\\
&\ge
(1-e^{-t})
\left[
H(I|Y^n)-(\alpha-1)nt
\right]
\end{align}
where \eqref{e_dominate} follows since $\Lambda_t$ dominates $T_{x^n,t}$, and
\eqref{e_notmuch} follows from \eqref{e_notm}.
\end{proof}

\section{Discussion}\label{sec_dis}
As mentioned in the introduction, recently Wu, Barnes and Ozgur \cite{WuBarnesOzgur} used the rearrangement inequalities to prove a tighter upper bound on the capacity of the Gaussian primitive relay channel which remains bounded away from the cutset bound also in the high relay rate regime (corresponding to the case of large $C_0$ or equivalently $\frac{1}{n}H(I|Y^n)$). (See \cite{ WuBarnesOzgur-general} for the treatment of the binary symmetric channel.)
Currently, the reverse hypercontractivity argument does not seem to be powerful enough in that regime.
In particular, the bound $$H(I|\bar{Y}^n)\le\frac{1}{1-e^{-2t}}H(I|X^n),$$ obtained in \eqref{e19} is looser than the trivial inequality $$H(I|\bar{Y}^n)\le H(I|X^n) + \sup_{P_{X^n}}I(I;X^n|\bar{Y}^n)$$ 
in the high entropy regime.
One possibility is that for highly symmetric measures (e.g.\ Gaussian), for which
a rearrangement inequality that characterizes the extremal sets/functions exists, the rearrangement approach is inherently
stronger than reverse hypercontractivity.
Another possibility is that there still exists certain semigroup argument that simplifies the proof \cite{WuBarnesOzgur}.
Indeed, we note that there exist semigroups versions of rearrangement inequalities which are known to imply (reverse) hypercontractive inequalities (see \cite[P117]{ledoux1996isoperimetry}). This remains an interesting open direction for future research.

\section{Acknowledgements}
Jingbo Liu would like to thank Professor Ramon van Handel and Sergio Verd\'u for their generous support, encouragement, and enlightening discussions during the time of this research.
In particular, the reference \cite[P117]{ledoux1996isoperimetry} mentioned in Section~\ref{sec_dis} was suggested by Ramon van Handel.
This work was supported by NSF grants CCF-1350595, CCF-1016625, CCF-0939370, CCF-1514538
and DMS-1148711, by ARO Grants W911NF-15-1-0479 and
W911NF-14-1-0094, and by the Center for Science of Information (CSoI), an NSF Science and Technology Center, under grant agreement CCF-0939370.

\bibliographystyle{ieeetr}
\bibliography{ref_maximal}

\end{document}